\documentclass[12pt,article]{memoir}
\usepackage[english.US]{babel}
\usepackage{amsthm}
\usepackage{amssymb}
\usepackage[scr]{rsfso}
\usepackage{upgreek}
\usepackage[leqno]{mathtools}
\usepackage{enumitem}

\usepackage{xr-hyper}
\usepackage{hyperref}

\usepackage{url}

\setsecheadstyle{\large\bfseries\memRTLraggedright}
\setsubsecheadstyle{\bfseries\memRTLraggedright}

\usepackage{todonotes}

\numberwithin{equation}{chapter}

\theoremstyle{plain}
\newtheorem{theorem}{Theorem}[chapter]
\newtheorem{proposition}[theorem]{Proposition}

\theoremstyle{definition}

\newtheorem{remark}[theorem]{Remark}

\theoremstyle{plain}

\theoremstyle{definition}
\newtoks{\thehRemark}
\newtheorem*{Remark}{\the\thehRemark}

\newenvironment{claim}[1][{\textup{(\theequation)}}]{\refstepcounter{equation}\vglue10pt
\begin{trivlist}
\item[{\hskip\labelsep#1}]}{\vglue10pt\end{trivlist}}

\newcommand{\Tr}{\operatorname{Tr}}

\newcommand{\supp}{\operatorname{supp}}

\newcommand{\blangle}{{\boldsymbol{\langle}}}
\newcommand{\brangle}{{\boldsymbol{\rangle}}}

\newcommand{\TF}{\mathsf{TF}}

\newcommand{\D}{\mathsf{D}}

\newcommand{\bR}{\mathbb{R}}
\newcommand{\bC}{\mathbb{C}}
\newcommand{\bS}{\mathbb{S}}

\newcommand{\x}{\mathsf{x}}
\newcommand{\y}{\mathsf{y}}

\newcommand{\sfH}{\mathsf{H}}

\newcommand{\sC}{\mathscr{C}}
\newcommand{\cH}{\mathcal{H}}
\newcommand{\sL}{\mathscr{L}}
\newcommand{\sH}{\mathscr{H}}

\newcommand{\fH}{\mathfrak{H}}

\newcommand{\cJ}{\mathcal{J}}

\newcommand{\cX}{\mathcal{X}}

\newcommand{\cZ}{\mathcal{Z}}

\begin{document}

\title{Upper Estimates for Electronic Density in Heavy Atoms and Molecules\thanks{\emph{2010 Mathematics Subject Classification}: 35P20, 81V70 .}\thanks{\emph{Key words and phrases}: electronic density, Thomas-Fermi approximation.}
}

\author{Victor Ivrii\thanks{This research was supported in part by National Science and Engineering  Research Council (Canada) Discovery Grant  RGPIN 13827}}

\maketitle

\begin{abstract}
We derive an upper estimate for electronic density $\rho_\Psi (x)$ in heavy atoms and molecules. While not sharp, on the distances $\gtrsim Z^{-1}$ from the nuclei it is still better than the known estimate $CZ^3$ ($Z$ is the total charge of the nuclei, $Z\asymp N$ the total number of electrons).
\end{abstract}

\chapter{Introduction}
\label{sect-1}
This paper is a result of my rethinking of three rather old but still remarkable papers \cite{HHT, S1, IaLS}, which I discovered recently.
The first of them derives the estimate electronic density $\rho_\Psi (x)$ from above via some integral also containing $\rho_\Psi $, the second  one provides an estimate $\rho_\Psi (x)=O(Z^3)$ where $Z$ is the total charge of nuclei and the third one derives the asymptotic of the averaged electronic density on the distances $O(Z^{-1})$ from the nuclei but its method works also on the larger distances.

The purpose of this paper is to provide a better upper estimate for $\rho_\Psi(x)$ on the distances larger than $Z^{-1}$ from the nuclei.

Let us consider the following operator (quantum Hamiltonian)
\begin{gather}
\mathsf{H}=\mathsf{H}_N\coloneqq   \sum_{1\le j\le N} H _{V,x_j}+\sum_{1\le j<k\le N}|x_j-x_k| ^{-1}
\label{eqn-1-1}\\
\shortintertext{on}
\fH= \bigwedge_{1\le n\le N} \sH, \qquad \sH=\sL^2 (\bR^3, \bC^q)
\label{eqn-1-2}\\
\shortintertext{with}
H_V =-\Delta -V(x)
\label{eqn-1-3}
\end{gather}
describing $N$ same type particles in (electrons) the external field with the scalar potential $-V$ (it is more convenient but contradicts notations of the previous chapters), and repulsing one another according to the Coulomb law.

Here $x_j\in \bR ^3$ and $(x_1,\ldots ,x_N)\in\bR ^{3N}$, potential $V(x)$ is assumed to be real-valued. Except when specifically mentioned we assume that
\begin{equation}
V(x)=\sum_{1\le m\le M} \frac{Z_m }{|x-\y_m|}
\label{eqn-1-4}
\end{equation}
where $Z_m>0$ and $\y_m$ are charges and locations of nuclei.

Mass is equal to $\frac{1}{2}$ and the Plank constant and a charge are equal to $1$ here. We assume that $N\asymp Z=Z_1+\ldots+Z_M$.

Our purpose is to a pointwise upper estimate for  the \emph{electronic density}
\begin{gather}
\rho_\Psi (x) =N\int |\Psi (x,x_2,\ldots,x_N)|^2\,dx_2\cdots dx_N.
\label{eqn-1-5}\\
\shortintertext{Let}
\ell(x)=\min_{1\le m\le M}|x-\y_m|
\label{eqn-1.6}
\end{gather}
is the distance to the nearest nucleus.  Our goal is  to prove  the following theorem:

\begin{theorem}\label{thm-1.1}
Let 
\begin{equation}
\min{1\le m<m'\le M}|\y_m-\y_{m'}|\ge Z^{-1/3+\sigma}
\label{eqn-1.7}
\end{equation}
with $\sigma> 0$. Then
\begin{enumerate}[label=(\roman*), wide, labelindent=0pt]
\item\label{thm-1.1-i}
For $\ell(x)\le Z^{-1/3}$ the following estimate holds:
\begin{equation}
\rho_\Psi (x)\le C\left\{\begin{aligned}
&Z^3  				&&\text{for\ \ } \ell (x)\le Z^{-8/9},\\
&Z^{19/9}\ell^{-1}         &&\text{for\ \ } Z^{-8/9}\le \ell (x)\le Z^{-7/9},\\
&Z^{197/90} \ell^{-9/10} 	&&\text{for\ \ } Z^{-7/9}\le \ell\le Z^{-1/3}.
\end{aligned}\right.
\label{eqn-1.8}
\end{equation}

\item\label{thm-1.1-ii}
For $\ell(x)\ge Z^{-1/3}$ the following estimate holds:
\begin{equation}
\rho_\Psi (x)\le C\left\{\begin{aligned}
&Z^{17/9}\ell ^{-9/5} 				&&\text{for\ \ } Z^{-1/3}\le \ell\le Z^{-5/18},\\
&Z^{19/9}\ell ^{-1} 	        \qquad &&\text{for\ \ } \ell \ge Z^{-5/18}.
\end{aligned}\right.
\label{eqn-1.9}
\end{equation}

\item\label{thm-1.1-iii}
Furthermore, if  $(Z-N)\ge C_0Z^{5/6}$ then
\begin{equation}
\rho_\Psi (x) \le Z^{19/9}\ell ^{-1} \qquad \text{for\ \ } \ell(x)\ge C_0(Z-N)_+^{-1/3}.
\label{eqn-1.10}
\end{equation} 
\end{enumerate}
\end{theorem}

\enlargethispage{1.5\baselineskip}

\begin{remark}\label{rem-1.2}
\begin{enumerate}[label=(\roman*), wide, labelindent=0pt]
\item\label{rem-1.2-i}
We would like to prove an estimate $\rho_\Psi (x)\le C\zeta^3$, or to discover that it does not necessarily hold.

\item\label{rem-1.2-ii}
We marginally improved our estimate (of the previous version) using \cite{ivrii:el-den-2}. We also added Statement~\ref{thm-1.1-iii}.
\end{enumerate}
\end{remark}

\textbf{Plan of the paper}.
In Section~\ref{sect-2} we prove a more subtle version of the main estimate of \cite{HHT}. In Section~\ref{sect-3} we provide upper estimates and asymptotics of  $\rho_\Psi$ integrated over small balls. In Section~\ref{sect-4} we study energy of electron-to-electron interaction (it involves a two-point correlation function) and in Section~\ref{sect-5} we prove upper estimates for $\rho_\Psi (x)$.

\chapter{Main intermediate inequality}
\label{sect-2}

We start from the main intermediate equality.

\begin{proposition}\label{prop-2.1}
Let $\Psi$ be an eigenfunction  of $\sfH_N$ with an eigenvalue $\lambda$. Let $\phi(|x|)$ be a real-valued spherically symmetric function.
Then
\begin{multline}
\rho _\Psi (0)= \\
(2\pi)^{-1}N \iint  (\partial_r K(x))\Psi (x,x_2,\ldots,x_N)\Psi^* (x,x_2,\ldots,x_N)\phi  (|x|)\,dxdx_2\cdots dx_N \\
-(8\pi)^{-1} \int\rho_\Psi(x)  \phi'''(|x|)\,dx,
\label{eqn-2.1}
\end{multline}
where 
\begin{equation}
K(x) \coloneqq -\sfH_N -(\partial_{r}^2 +2r^{-1}\partial_{r})
\label{eqn-2.2}
\end{equation}
is an operator in the auxiliary space 
$\cH\coloneqq \bigotimes_{n=2,\ldots,N} \sL^2(\bR^3,\bC^q) \otimes \bC^q$ with an inner product $\langle.,.\rangle$, $\phi'''(r)=\partial_r^3\phi(r)$ and 
$x=(r,\theta)\in \bR^+\times \bS^2$.
\end{proposition}

\begin{proof}
Let us consider $\Psi$ as a function of $x\in \bR^3$ with values in  the auxiliary space $\cH$,  and and let $u=r \Psi$ where $(r,\theta)$ are spherical coordinates in $\bR^3$. Then similar to (9) of \cite{HHT}
\begin{align}
\langle \Psi (0),\Psi(0)\rangle  =& -(2\pi)^{-1}\int \langle \partial_r u,\partial_r ^2u \rangle \phi(r)r^{-2}\,dx \label{eqn-2.3}\\
&-(4\pi)^{-1}\int \langle \partial_r u,\partial_r u \rangle \phi'(r)r^{-2}\,dx\notag
\end{align}
and since $r^{-1}\partial_r^2 u = r\Delta_r\Psi \coloneqq r(\partial_r^2 +2r^{-1}\partial_r)\Psi $, the first term on the right is equal to
\begin{multline}
-(2\pi)^{-1} \iint \langle  \partial_ru,r\Delta_r \Psi  \rangle \phi(r)\,drd\theta\\
=\pi ^{-1} \iint \langle  \partial_ru,(K+\lambda) u\rangle \phi(r) drd\theta=\\
=-(2\pi)^{-1}\int \langle  \Psi,K' \Psi\rangle \phi(r)\, dx -(2\pi)^{-1}\iint  \langle \Psi ,(K+\lambda)\Psi \rangle \phi'(r)\,dx
\label{eqn-2.4}
\end{multline}
because $\Delta _r\Psi =-2(K+\lambda)\Psi$, where $-K$ is the rest of multiparticle Hamiltonian (including $-r^{-2}\Delta_\theta$) and we integrated by parts.

The first term in the latter formula is a corresponding term in \cite{HHT}, albeit truncated with $\phi$, and we have new terms
\begin{equation*}
-(4\pi)^{-1}\iint \langle \partial_r (r\Psi),\partial_r (r\Psi)\rangle \phi'(r)\,drd\theta  - 
(2\pi)^{-1} \int \langle \Psi ,(W+\lambda)\Psi \rangle\phi'(r)\,dx\,.
\end{equation*}
Integrating by parts the first term we get
\begin{equation*}
(4\pi)^{-1}\iint \blangle  r\Psi,\partial^2_r (r\Psi)\brangle \phi'(r)\,drd\theta  +
 (4\pi)^{-1}\iint \blangle  r\Psi,\partial_r (r\Psi)\brangle \phi''(r)\,drd\theta,
\end{equation*}
where the first term cancels with the second term in (\ref{eqn-2.4}), while the second term integrates by parts one more time resulting in
the last term in (\ref{eqn-2.1}). \end{proof}

Applying (\ref{eqn-2.1}) to our problem, and using skew-symmetry of $\Psi$, we get
\begin{multline}
\rho_\Psi (0) = (2\pi)^{-1}\int  \sum_{m} Z_m  \frac{x \cdot (x -\y_m)} {|x |\cdot |x -\y_m|^3}\rho_\Psi(x)\phi (|x|)\,dx \\
+ (2\pi)^{-1}N(N-1)  \int\underbracket{\frac{x_1\cdot (x_2-x_1)}{|x_1|\cdot |x_2-x_1|^3}}|\Psi (x_1,x_2,\ldots,x_N)|^2
\underbracket{\phi(|x_1|)}\, dx_1 \cdots dx_N\\
- (2\pi)^{-1} N \int |x |^{-3} |\nabla_{\theta} \Psi (x ,x_2,\ldots,x_N)|^2\phi (|x|) \, dx dx_2 \cdots dx_N\\ 
-(8\pi)^{-1} \int\rho_\Psi(x)  \phi ''' (|x|)\,dx.
\label{eqn-2.5}
\end{multline}

Symmetrizing the second term with respect to $x_1$ and $x_2$ we instead of the product of two indicated factors will get 
\begin{align*}
&\frac{1}{4} \bigl(\phi  (|x_1|)+\phi  (|x_2|)\bigr) 
\Bigl(\frac{x_1\cdot(x_2-x_1)}{|x_1|\cdot |x_2-x_1|^3}-\frac{x_2\cdot (x_2-x_1)}{|x_2| \cdot |x_2-x_1|^3}\Bigr) \\
+&\frac{1}{4} \bigl(\phi  (|x_1|)-\phi  (|x_2|)\bigr) 
\Bigl(\frac{x_1\cdot (x_2-x_1)}{|x_1|\cdot |x_2-x_1|^3}+\frac{x_2\cdot (x_2-x_1)}{|x_2|\cdot |x_2-x_1|^3}\Bigr)
\end{align*}
with the big parenthesis on the first line equal to
\begin{align}
-\frac{|x_1|+|x_2|}{|x_1-x_2|^3}\Bigl (1-\frac{x_1\cdot x_2}{|x_1|\cdot |x_2|}\Bigr)
\label{eqn-2.6}\\
\intertext{and the big parenthesis on the second line equal to}
-\frac{|x_1|-|x_2|}{|x_1-x_2|^3}\Bigl (1+\frac{x_1\cdot x_2}{|x_1|\cdot |x_2|}\Bigr)\,.
\label{eqn-2.7}
\end{align}
One can see that the former is negative, and the latter, multiplied by $\bigl(\phi  (|x_1|)-\phi(|x_2|)\bigr)$, is non-negative if $\phi$ is non-decreasing function.
Let us shift the origin to point $\x$ and observe that the first term in (\ref{eqn-2.5}) is equal to
\begin{equation}
(2\pi)^{-1} \int \sum_m Z_m \frac{(x-\x)\cdot (x-\y_m)}{|x-\x|\cdot |x-\y_m|^3}\rho_\Psi(x)\phi (|x-\x|)\,dx\,.
\label{eqn-2.8}
\end{equation}

Consider first case $\phi=1$. Then we get 
\begin{equation}
\rho_\Psi (\x) \le   (2\pi)^{-1} \int \sum_m Z_m |x-\y_m|^{-2}\rho_\Psi(x)\,dx.
\label{eqn-2.9}
\end{equation}
Indeed, the second term in the right-hand expression of (\ref{eqn-2.5}) is non-positive due to above analysis analysis, so is the third term, and the fourth term vanishes while the first term does not exceed the right-hand expression

Applying Proposition~\ref{prop-3.1}  below we arrive to the following estimate
\begin{equation}
\rho_\Psi (x) \le CZ^3.
 \label{eqn-2.10}
\end{equation}

In the general case  we arrive to

\begin{proposition}\label{prop-2.2}
In the framework of Proposition~\ref{prop-2.1}
\begin{multline}
\rho_\Psi (\x) \le (2\pi)^{-1} \int \sum_m Z_m \frac{(x-\x)\cdot (x-\y_m)}{|x-\x|\cdot |x-\y_m|^3}\rho_\Psi(x)\phi (|x-\x|)\,dx\\
\begin{aligned}
&+Ct^{-1}\iint _{B(\x,t)\times B(\x,t)} |x-y|^{-1}\rho^{(2)}_\Psi (x,y)\,dxdy \\
&+ C\iint _{B(\x,t)\times (\bR^3\setminus B(\x,t))} |y-\x|^{-2}\rho^{(2)}_\Psi (x,y)\,dxdy\,,
\end{aligned}
 \label{eqn-2.11}
 \end{multline}
where
\begin{equation}
\rho ^{(2)}_\Psi (x,y)\coloneqq   N(N-1)\int |\Psi (x,y,x_3,\ldots,x_N)| ^2\, dx_3\cdots dx_N 
\label{eqn-2.12}
\end{equation}
is a two-point correlation function.
\end{proposition}

Recall that  
\begin{equation}
\int \rho ^{(2)}_\Psi (x,y)dy =(N-1)\rho _\Psi (x).
\label{eqn-2.13}
\end{equation}

\begin{remark}\label{rem-2.3}
\begin{enumerate}[label=(\roman*), wide, labelindent=0pt]
\item\label{rem-2.3-i} 
Inequality (\ref{eqn-2.9}) for $M=1$ and $\x=\y_1$ is the main result of \cite{HHT}. Our main achievement so far is an introduction of the truncation $\phi$. However it brings three new terms in the right-hand expression of the estimate.
\item\label{rem-2.3-ii}
Estimate (\ref{eqn-2.10}) (with a specified albeit not sharp constant) was proven in \cite{S1} for $x=\y_m$.
\item\label{rem-2.3-iii}
This estimate definitely has a correct magnitude as $|x-\y_m|\lesssim Z^{-1}$ and $Z_m\asymp Z$.
\end{enumerate}
\end{remark}

\chapter{Estimates of the averaged electronic density}
\label{sect-3}

We will need the following estimate (\ref{strongscott-eqn-3.3})  from \cite{ivrii:strong-scott}:
\begin{equation}
\int U\rho_\Psi \,dx  \le 
\Tr (H_{W+\nu}^-)   -  \Tr (H_{W+U+\nu}^-) +C Z^{5/3-\delta}
\label{eqn-3.1}
\end{equation}
with $\delta=\delta(\sigma)$, $\delta>0$ for $\sigma>0$ and $\delta=0$ for $\sigma=0$.

First, we use this estimate in the very rough form:

\begin{proposition}\label{prop-3.1}
The following estimate holds:
\begin{equation}
\int |x-\y_m|^{-2}\rho_\Psi (x)\,dx \le CZ^2.
\label{eqn-3.2}
\end{equation}
\end{proposition}

\begin{proof}
Let $\psi(x)$, $\psi_0(x)$ be  cut-off functions, $\psi(x)=0$ in $\{x\colon \ell(x)\le b\}$, $\psi_0(x)=0$ in $\{x\colon \ell(x)\ge 2b\}$, $\psi+\psi_0=1$, $ b= Z^{-1}$.  Then
\begin{equation}
\Tr (H_{W+U+\nu}^-)=\Tr (H_{W+U+\nu}^-\psi_0) + \Tr (H_{W+U+\nu}^-\psi)
 \label{eqn-3.3}
\end{equation}
Using the semiclassical methods of \cite{monsterbook:5}, Section \ref{monsterbook-sect-25-4} in the simplest form, we conclude that for 
 $U= \epsilon |x-\y_m|^{-2}$ the second term on the right (with an opposite sign) could be replaced by its Weyl approximation 
\begin{equation}
-\frac{2}{5}\kappa \int (W+U+\nu)^{5/2}_+ \psi (x)\,dx
\label{eqn-3.4}
\end{equation}
with an error not exceeding $CZ^2$ where here and below $\kappa=q/(6\pi^2)$. The same is true for $U=0$. One can see easily that the difference between expression (\ref{eqn-3.4}) and the same expression for $U=0$ does not exceed
\begin{equation}
C\int \bigl[(W+\nu)_+^{3/2} U + U^{5/2} \bigr]\,dx ,
\label{eqn-3.5}
\end{equation}
which does not exceed $CZ^2$. 

Consider the first term in the right-hand expression of (\ref{eqn-3.3}). Using variational methods of \cite{monsterbook:5}, Section \ref{monsterbook-sect-9-1} we can reduce it to the analysis of the same operator in $\cX^0\coloneqq \{x\colon \ell(x)\le 4b\}$ with the Dirichlet boundary conditions on $\partial X$. Observing that eigenvalue counting function for such operator is $O(1+\lambda ^{3/2}Z^{-3})$ 
(for $\epsilon$ sufficiently small), we conclude that the first term in (\ref{eqn-3.3}) also does not exceed $CZ^2$. Estimate (\ref{eqn-3.2}) has been proven.
\end{proof}

Let us return to (\ref{eqn-3.1}) and consider $U=   \zeta^2 \phi_t (x;\x)$ where $\x$ is a fixed point with  
\begin{align}
&\ell(x)\coloneqq \min_m |x-\y_m|\ge Z^{-1}
\label{eqn-3.6},\\
&\zeta (x)\coloneqq  \max \bigl(Z^{1/2}\ell(x)^{-1/2},\,\ell(x)^{-2}\bigr)
\label{eqn-3.7}
\end{align}
and $\phi_t (x;\x) = \phi _0 (t^{-1}|x-\x|)$, $\phi \in \sC_0^\infty ([-1,1])$, $0\le \phi \le 1$. We assume that
\begin{equation}
\zeta^{-1}\le t \le \frac{\ell }{2}
\label{eqn-3.8}
\end{equation}
with $\ell=\ell(\x)$, $\zeta=\zeta(\x)$,  where the last inequality allows us to apply semiclassical methods. Consider with $0\le \varsigma \le 1$
\begin{multline}
\Tr (H_{W+\nu}^-)   -  \Tr (H_{W+\varsigma U+\nu}^-)\\
=\int_0^\varsigma  \Tr \Bigl( U \bigl[ \uptheta(-H_{W+sU +\nu}) - \uptheta(-H_{W+\nu})\bigr]\Bigr)\,ds
\label{eqn-3.9}
\end{multline}
and apply semi-classical method to the right-hand expression. Then we get
\begin{multline}
\Tr (H_{W+\nu}^-)   -  \Tr (H_{W+\varsigma  U+\nu}^-)\\
\begin{aligned}
&=\kappa \int \int_0^\varsigma  \Bigl(U\bigl[(W+sU+\nu)_+^{3/2} - (W+\nu )_+^{3/2} \bigr]\Bigr)\,dxds
+ O(\varsigma \zeta^4t^2)\\
&= \frac{2}{5}\kappa \int\Bigl((W+\varsigma U+\nu)_+^{5/2} - (W+\nu )_+^{5/2}\Bigr)\,dx + O(\varsigma \zeta^4t^2).
\end{aligned}
\label{eqn-3.10}
\end{multline}

Indeed, factor $U$ is $O(\zeta^2)$ and therefore the semiclassical error is $O(\zeta^4t^2)$ since the effective semiclassical parameter is $h=(\zeta t)^{-1}$. Observe that the principal part in the right-hand expression does is  $O(\varsigma \zeta^5t^3)$.

Then after division by $ \varsigma \zeta^2$ (\ref{eqn-3.1}) becomes
\begin{multline}
 \int \phi_t (x;\x)\rho_\Psi (x)\,dx \le  \int \phi_t (x;\x)\rho  (x)\,dx \\
 +  C\Bigl(\zeta^2t^2 + \varsigma \zeta^3t^3 + \varsigma^{-1} Z^{5/3-\delta}\Bigr)\,.
 \label{eqn-3.11}
\end{multline}
Replacing $\phi_t (x;\x)$ by $-\phi_t (x;\x)$ in this inequality and minimizing by $\varsigma\in (0,1]$  we arrive to the first statement of the following proposition:

\begin{proposition}\label{prop-3.2}
\begin{enumerate}[label=(\roman*), wide, labelindent=0pt]
\item\label{prop-3.2-i}
Under assumptions \textup{(\ref{eqn-3.6})}--\textup{(\ref{eqn-3.8})}
\begin{multline}
|\int ( \rho_\Psi (x)-\rho(x)) \phi_t (x;\x)\,dx |\\
\le 
C\Bigl(\zeta^2t^2 +  \zeta^{1/2} t^{3/2} Z^{5/6-\delta/2} + \zeta^{-2}Z^{5/3-\delta}\Bigr)\,.
\label{eqn-3.12}
\end{multline}
\item\label{prop-3.2-ii}
Further,
\begin{equation}
|\int  \rho_\Psi (x)\phi_t (x;\x)\,dx |\le 
C\Bigl(\zeta^3t^3 +  t^{6/5}Z^{1-3\delta/5}\Bigr)\,.
\label{eqn-3.13}
\end{equation}
\item\label{prop-3.2-iii}
Furthermore, if $N<Z$ then
\begin{equation}
|\int  \rho_\Psi (x)\phi_t (x;\x)\,dx |\le Ct^{6/5}Z^{1-3\delta/5}\qquad 
\text{for\ \ } \ell (\x) \ge C_0(Z-N)_+^{-1/3}.
\label{eqn-3.14}
\end{equation}
\end{enumerate}
\end{proposition}

To prove the second statement, we consider $\varsigma >0$ (without restriction $\varsigma \le 1$); then instead of (\ref{eqn-3.11}) we have 
\begin{equation}
 \int \phi_t (x;\x)\rho_\Psi (x)\,dx \le  
 C\Bigl(\zeta^3t^3 +\varsigma ^{3/2}\zeta^3t^3 +  \varsigma^{-1}\zeta^{-2} Z^{5/3-\delta}\Bigr)
 \label{eqn-3.15}
\end{equation}
and we optimize it by $\varsigma>0$.

The third statement follows from the same arguments and the fact that  recall that $\rho^\TF(x)=0$ for 
$\ell (x) \ge C_0(Z-N)_+^{-1/3}$ and therefore (\ref{eqn-3.15}) holds without the first term in the right-hand expression.

\chapter{Estimates of the correlation function}
\label{sect-4}

We will need the following Proposition~\ref{monsterbook-prop-25-5-1} from \cite{monsterbook:5} (first proven in \cite{ruskai:solovej}):

\begin{proposition}\label{prop-4.1}
Let $\theta \in \sC^\infty(\bR^3)$, such that 
\begin{equation}
0\le \theta \le 1.
\label{eqn-4.1}
\end{equation}
Let $\chi \in \sC^\infty(\bR^6)$ and 
\begin{multline}
\cJ =|\int \Bigl(\rho _\Psi ^{(2)}(x,y)-
\rho(y)\rho _\Psi (x)\Bigr)\theta (x)\chi (x,y)\,dxdy |\le \\[3pt]
C\sup_x \|\nabla _y\chi \|_{\sL ^2(\bR_y ^3)}
\Bigl((Q+\varepsilon ^{-1}N+T)^{\frac{1}{2}}\Theta + P^{\frac{1}{2}}\Theta^{\frac{1}{2}}\Bigr) +
C \varepsilon N\|\nabla _y\chi \|_{\sL ^\infty }\Theta
\label{eqn-4.2}
\end{multline}
with 
\begin{gather}
Q=\D(\rho_\Psi-\rho^\TF,\, \rho_\Psi-\rho^\TF),
\label{eqn-4.3}\\
\Theta =\Theta_\Psi \coloneqq  \int \theta (x)\rho _\Psi (x)dx,
\label{eqn-4.4}\\
T= \sup _{\supp(\theta)} W,
\label{eqn-4.5}\\
P= \int |\nabla \theta^{\frac{1}{2}}|^2 \rho_\Psi \,dx.
\label{eqn-4.6}
\end{gather}
respectively and arbitrary
$\varepsilon \le Z^{-\frac{2}{3}}$.
\end{proposition}

We cannot apply it directly to estimate the second to the last term in (\ref{eqn-2.11}) because of singularities. Let us consider
\begin{equation}
\iint_{\bR^3\times \bR^3}  |x-y|^{-1}\rho^{(2)}_\Psi (x,y)\,dxdy .
\label{eqn-4.7} 
\end{equation}
Let us make an $\ell$-admissible partition of unity $\phi_\iota$ in  with $\ell$-admissible $\phi^{1/2}_\iota$. We set $\ell(x) =Z^{-1}$ if  $|x-\y_m|\le Z^{-1}$.  
Let us consider first
\begin{equation}
\int |x-y|^{-1}\rho^{(2)}_\Psi (x,y)\phi_\iota (x)\phi_\varkappa(y) \,dxdy
\label{eqn-4.8} 
\end{equation}
in the case of $\phi_\iota$ and $\phi_\kappa$ having disjoint supports. Without any loss of the generality we can consider 
$\ell_x \le \ell_y$, where subscripts $x,y$ are referring to supports of $\phi_\iota$, $\phi_\kappa$ respectively.

Let $\theta(x)=\phi_\iota(x)$ and
\begin{equation}
\chi(x,y)=\bar{\chi}(x,y)\coloneqq |x-y|^{-1} \bar{\phi}_\iota (x)\phi_\kappa(y)
\label{eqn-4.9}
\end{equation}
where 
$\bar{\phi}_\iota$ which are $\ell$-admissible and equal $1$ in the $\ell$-vicinity of $\supp(\phi_\iota)$. 
Then $T= \zeta_x$ and for $\ell_x \le Z^{-5/21}$  in virtue of Proposition~\ref{prop-3.2}\,\footnote{\label{foot-1} Indeed,  $\zeta^3\ell^3 \ge Z^{5/3-\delta}\zeta^{-2}\iff \ell \le Z^{-5/21+\delta/7}$.}
\begin{gather}
\Theta_\Psi \asymp \zeta_x^3\ell_x^3, \qquad P\asymp \zeta_x^3\ell_x
\label{eqn-4.10}\\
\shortintertext{and}
\|\nabla_y \chi \|_{\sL^2 (\bR^3_y))}\asymp  d_{x,y}^{-1} \ell_y^{1/2} , \qquad 
\|\nabla \chi\|_{\sL^\infty}\asymp d_{x,y}^{-1}\ell_y^{-1}
\label{eqn-4.11}
\end{gather}
where $d_{x,y}\ge \ell_y$ is the distance between  supports of $\phi_\iota$ and $\phi_\varkappa$. Then the right-hand expression of (\ref{eqn-4.2}) is
\begin{equation*}
C  \zeta_x^3\ell_x^3\Bigl( \ell_y^{-1/2}(Z^{5/6} + \zeta_x)  +\varepsilon^{-1/2}\ell_y^{-1/2} Z ^{1/2} 
+  \varepsilon Z  \ell_y^{-2} +\ell_y^{-1/2} \ell_x^{-1} \Bigr)
\end{equation*}
and minimizing by $\varepsilon \le Z^{-2/3}$ we get
\begin{equation*}
C \zeta_x^3\ell_x^3\Bigl( \ell_y^{-1/2}(Z^{5/6-\delta} + \zeta_x)  +  Z ^{2/3} \ell_y^{-1} +
  \ell_y^{-1/2} Z ^{5/6} +\ell_y^{-1/2} \ell_x^{-1} \Bigr).
\end{equation*}

Observe that all powers of $\ell_y$ are  negative. Therefore summation over all elements of $y$-partition results in the same expression albeit with $\ell_y$ replaced by $\ell_x=\ell$:
\begin{equation*}
C \zeta^3\ell ^2\Bigl( \ell ^{1/2}(Z^{5/6} + \zeta)  +  Z ^{2/3}  + \ell^{-1/2} + \ell^{1/2} Z ^{5/6} \Bigr).
\end{equation*}

For $\ell \le Z^{-1/3}$ we have $\zeta =Z^{1/2}\ell ^{-1/2}$ and all powers are positive with the exception of one term, where the power is $0$, 
and for $\ell\ge Z^{-1/3}$ we have $\zeta=\ell^{-2}$ and all powers are negative. Therefore summation over all elements of $x$-partition results in the same expression albeit with $\ell= Z^{-1/3}$, $\zeta=Z^{2/3}$, with the exception of one term which gains a logarithmic factor. We get $CZ^2$. Then
\begin{equation}
|\sum _{\iota, \varkappa} \iint \Bigl(\rho^{(2)}_\Psi (x,y)- \rho (y)\rho_\Psi (x)\Bigr)\phi_{\iota}(x)
\phi_{\kappa}(y),dxdy|\le CZ^2
\label{eqn-4.12}
\end{equation}
with summation over indicated pairs of elements of the partition 
(disjoint, with $\ell_x\le \min(Z^{-5/21+\delta/7},\, \ell_y)$ ).

Let us prove that 

\begin{claim}\label{eqn-4.13}
Estimate (\ref{eqn-4.12}) also holds with $\rho_\Psi(x)$ replaced by $\rho(x)$ and therefore it holds for a sum over paits of elements with 
$\min (\ell_x,\, \ell_y)\le \ell^*=Z^{-5/21+\delta/7}$. 
\end{claim}

Indeed, in virtue of the proof of Proposition~\ref{prop-3.2} (before minimizing by $\varsigma$) the  error 
\begin{equation}
|\iint \bigl(\rho_\Psi (x)-\rho (x) \bigr)\rho(y)\phi_{\iota}(x)\phi_{\kappa}(y),dxdy|
\label{eqn-4.14}
\end{equation}
on each pair of elements does not exceed  $C   \zeta_y^3 \ell_y^2$
with all powers of $\ell_y$ positive for $\ell _y\le Z^{-1/3}$ and negative for $\ell_y\ge Z^{-1/3}$. Then summation with respect to $y$-partition (recall, that $\ell_y\ge \ell_x$) results in
\begin{equation*}
C\bigl(\zeta _x^2 \ell_x^2 + \varsigma \zeta _x^3 \ell_x^3 +\varsigma^{-1} \zeta_x^{-2} Z^{5/3-\delta}\bigr)\times  
\left\{\begin{aligned}
&Z^{4/3} &&\text{for\ \ }  \ell_x\le Z^{-1/3},\\
&\zeta_x^3 \ell_x^2 &&\text{for\ \ } \ell_x\ge Z^{-1/3},
\end{aligned}\right.
\end{equation*}
with the first line corresponding to $\ell_y=Z^{-1/3}$, $\zeta_y= Z^{2/3}$ and the second line corresponding to 
$\ell_y=\ell_x$, $\zeta_y= \zeta_x$. 

Powers of $\ell_x$ are positive for $\ell_x\le Z^{-1/3}$  and negative for $\ell_x\ge Z^{-1/3}$,  and summation with respect to $x$-partition results in the value as $\ell_x=Z^{-1/3}$, $\zeta_x= Z^{2/3}$, which is 
\begin{equation*}
CZ^2 +C\varsigma Z^{7/3}+ C\varsigma^{-1}  Z^{5/3-\delta}.
\end{equation*}
Minimizing by $\varsigma = Z^{-1/3}$ we conclude that the sum of expressions (\ref{eqn-4.13}) over required pairs does not exceed $CZ^2$, which in turn implies (\ref{eqn-4.12}).
  
Consider now the case when supports of elements are not disjoint. Then we take
\begin{equation}
\chi(x,y)=\bar{\chi}(x,y) \eta (|x-y|/s)= |x-y|^{-1} \bar{\phi}_\iota (x)\phi_\kappa(y) \eta (|x-y|/s)
\label{eqn-4.15}
\end{equation}
with $\eta(t)$ smooth function, equal $0$ at $(0,\frac{1}{2})$ and $1$ at $(1,\infty)$; $s \le Z^{-1/3} $ will be selected later\footnote{\label{foot-2} Since in this case $\ell_x=\ell_y$ and $\zeta_x=\zeta_y$ we skip subscripts.}. Then while (\ref{eqn-4.10}) is preserved, (\ref{eqn-4.11}) should be replaced by 
\begin{equation}
\|\nabla_y \chi \|_{\sL^2 (\bR^3_y))}\asymp s^{-1/2}  , \qquad 
\|\nabla_y \chi\|_{\sL^\infty}\asymp s^{-2}.
\label{eqn-4.16}
\end{equation}
Then the right-hand expression (\ref{eqn-4.2}) is
\begin{equation*}
C  s^{-1/2} \zeta ^3 \ell ^3 \Bigl(   (Z^{5/6} + \zeta )  +\varepsilon^{-1/2} Z ^{1/2} 
 + \varepsilon s^{-3/2} Z   + \ell ^{-1}\Bigr)\,,
\end{equation*}
and minimizing by $\varepsilon \le Z^{-2/3}$ we get 
\begin{multline}
|\int \phi_\iota(x) \chi(x,y) \bigl( \rho^{(2)}_\Psi(x,y)-\rho_\Psi(x)\rho(y)   \bigr)\,dxdy |\\
\le
Cs^{-1/2}   \zeta ^3 \ell^3 \Bigl(  Z^{5/6} + \zeta   + s^{-1/2} Z^{2/3} + \ell ^{-1}   \Bigr).
\label{eqn-4.17}
\end{multline}
Note that summation of (\ref{eqn-4.17}) over partition returns its value as $\ell=Z^{-1/3}$, namely,
$Cs^{-1} Z^{5/3}$.

Consider for $t\colon s\le t \le \ell$ zone $\{(x,y)\colon |x-y|\asymp y\}$ and make there $t$-admissible subpartition with respect to $x$, $y$. Then contribution of each pair of subelements to 
\begin{gather*}
 |\int \phi_\iota(x) \chi(x,y) \bigl( \rho_\Psi(x)-\rho(x)\bigr) \rho(y) \,dxdy |\\
\shortintertext{does not exceed}
C \bigl(\zeta^2 t^2 + \varsigma \zeta^3t^3 + \varsigma^{-1}\zeta^{-2} Z^{5/3-\delta}\bigr) \zeta^3t^2\\
\intertext{and since there are $\asymp \ell^{3}t^{-3}$ of such pairs, we get}
C \bigl(\zeta^2t^2 + \varsigma \zeta^3t^3 + \varsigma^{-1}\zeta^{-2}Z^{5/3-\delta}\bigr)\zeta^3\ell^3 t^{-1}.
 \end{gather*} 

Then summation over $t\colon s\le t \le \ell$ returns
\begin{gather*}
C \bigl(\zeta^2\ell  + \varsigma \zeta^3\ell^2  + \varsigma^{-1}s^{-1}\zeta^{-2}Z^{5/3-\delta}\bigr)\zeta^3\ell^3 \\
\intertext{and summation over over partition returns its value as $\ell=Z^{-1/3}$, namely}
C \bigl(Z^2  + \varsigma Z^{7/3}  + \varsigma^{-1}s^{-1}Z^{4/3-\delta}\bigr).
\end{gather*}
Minimizing by $\varsigma= (sZ)^{-1/2}$ we get 
$C \bigl(Z^2  +   s^{-1/2}Z^{11/6-\delta}\bigr)$.

On the other hand, 
\begin{equation*}
\iint \phi_\iota (x)\bigl(\bar{\chi}(x,y)-\chi (x,y)\bigr) \rho (x)\rho(y)\,dxdy \asymp s^2\zeta^6  \ell ^3 \,,
\end{equation*}
and summation over $\ell \ge Z^{-1/3}$ returns its value at $Z^{-1/3}$, which is $Cs^2Z^3$, but summation over $\ell\le Z^{-1/3}$ returns $s^2Z^3\log Z$. To remedy this we replace for $\ell\le Z^{-1/3}$ constant $s$ by 
$s_x= s(\ell_x Z^{1/3})^{\delta'}$ with small $\delta'>0$. It will not affect our previous estimates.

Consider the sum of these three right-hand expressions 
\begin{equation*}
C   \bigl(Z^2+  s^{-1}Z^{5/3}+ s^{-1/2}Z^{11/6-\delta} +s^2 Z^3\bigr)
\end{equation*}
and minimize it by $s$; we get $CZ^{19/9}$ achieved as $s=Z^{-4/9}$.

Since we want $s\le \ell$ we finally set
\begin{equation}
s_x=\left\{\begin{aligned}
&Z^{-4/9} &&\text{for\ \ } \ell_x\ge Z^{-1/3},\\
&\min (Z^{-4/9}(\ell_x Z^{1/3})^{\delta'},\, \ell_x)&&\text{for\ \ }\ell_x\le Z^{-1/3}.
\end{aligned}\right.
\label{eqn-4.18}
\end{equation}

Observe that 
\begin{gather}
\iint_{\{x\colon \ell_x \ge Z^{-5/21}\} } |x-y|^{-1}\rho(x)\rho(y) \asymp Z^{41/21}
\notag\\
\intertext{and we arrive to}
\iint_{\Omega } |x-y|^{-1}\bigl(\rho^{(2)}_\Psi (x,y)-\rho(x)\rho(y)\bigr)\,dx dy\le CZ^{19/9}
\label{eqn-4.19}\\
\shortintertext{and}
\iint_{\bR^3_x\times \bR^3_y \setminus \Omega } |x-y|^{-1}\rho(x)\rho(y)\,dx dy\le CZ^{19/9}
\label{eqn-4.20}\\
\shortintertext{with}
\Omega =\{(x,y)\colon \ell_x \le Z^{-5/21}, \, |x-y|\ge s_y\}.
\label{eqn-4.21}
\end{gather}

Therefore
\begin{equation}
\int_{\Omega }  |x-y|^{-1} \rho_\Psi^{(2)} (x,y) \,dxdy \ge 
\int _{\bR^6}|x-y|^{-1} \rho(x)\rho(y)\,dxdy - CZ^{19/9}.
\label{eqn-4.22}
\end{equation}

However we know that (see, f.e. Section~\ref{monsterbook-sect-25-2} of \cite{monsterbook:5} )
\begin{gather}
E_N \ge  \Tr \bigl((H_W-\nu)^-\bigr)  -\D(\rho_\Psi, \rho) +
\frac{1}{2} \int |x-y|^{-1}\rho^{(2)} _\Psi (x,y)\,dxdy
\label{eqn-4.23}\\
\shortintertext{and}
E_N \le \Tr \bigl((H_W-\nu)^-\bigr) - \frac{1}{2} \D(\rho,\rho) + CZ^{5/3}
\label{eqn-4.24}
\end{gather}
with $\rho=\rho^\TF$, $W=W^\TF$, $\D(f,g)\colon \iint |x-y|^{-1}f(x)g(x)\,dx$. Then
\begin{gather}
\int |x-y|^{-1}\rho^{(2)} _\Psi (x,y)\,dxdy \le 2\D(\rho_\Psi-\rho ,\rho) + \D(\rho,\rho) + CZ^{5/3}
\notag\\
\intertext{and from}
|\D(\rho_\Psi -\rho,\rho)| \le \D(\rho_\Psi -\rho,\rho_\Psi -\rho)^{1/2}\D(\rho,\rho)^{1/2}\le CZ^{5/6}\times Z^{7/6}=CZ^2
\notag\\
\shortintertext{we conclude that} 
\int |x-y|^{-1}\rho^{(2)} _\Psi (x,y)\,dxdy \le \D(\rho,\rho) + CZ^{2}.
\label{eqn-4.25}
\end{gather}

Combining with (\ref{eqn-4.22}) we conclude that
\begin{equation}
\int_{\cZ}  |x-y|^{-1} \rho_\Psi^{(2)} (x,y) \,dxdy \le   CZ^{19/9}
\label{eqn-4.26}
\end{equation}
for $\cZ=\bR^3_x\times \bR^3_y \setminus \Omega$.

\chapter{Proof of Theorem~\ref{thm-1.1}}
\label{sect-5}

Now in the last two terms 
\begin{align}
&Ct^{-1}\iint _{B(\x,t)\times B(\x,t)} |x-y|^{-1}\rho^{(2)}_\Psi (x,y)\,dxdy 
\notag\\
+& C\iint _{B(\x,t)\times (\bR^3\setminus B(\x,t))} |y-\x|^{-2}\rho^{(2)}_\Psi (x,y)\,dxdy,
\notag\\
\intertext{in (\ref{eqn-2.11}) we replace $\rho^{(2)}_\Psi (x,y)$ by $\rho(x)\rho(y)$ and get}
&Ct^{-1} \iint _{B(\x,t)\times B(\x,t)} |x-y|^{-1} \rho(x) \rho(y)\,dxdy 
\label{eqn-5.1}\\
+&C \iint _{B(\x,t)\times (\bR^3\setminus B(\x,t))} |\x-y|^{-2} \rho(x) \rho(y)\,dxdy\notag\\
\intertext{and the first term does not exceed $C\zeta^6 t^4$, while the second term does not exceed}
&C\zeta^3 t^3 \int_{\bR^3\setminus B(\x,t)} |\x-y|^{-2} \rho(y)\,dy.
\label{eqn-5.2}
\end{align}
The largest error comes from the first term  when integral is taken over $B(\x,t)\times B(\x,t)\cap \cZ$ and in virtue of of (\ref{eqn-4.26}) it does not exceed $Ct^{-1}Z^{19/9}$, all other errors are lesser (to prove it we need just to repeat arguments of the previous section).

Observe that for $\ell\coloneqq \ell_\x \le Z^{-1/3}$ the largest contribution to the integral in (\ref{eqn-5.2}) comes from the layer 
$\{y\colon \ell_y\asymp \ell_\x\}$ and it is of magnitude $\zeta^3 \ell_x$. On the other hand, for $\ell_\x \ge Z^{-1/3}$  the largest contribution to the integral in (\ref{eqn-5.1}) comes from the layer 
$\{y\colon \ell_y\asymp Z^{-1/3}\}$ and it is of magnitude $Z \ell^{-2}$; the first term in (\ref{eqn-5.1}) is smaller.

Therefore we estimate two last terms in (\ref{eqn-2.11}) by
\begin{equation}
Ct^{-1} Z^{19/9} + C\left\{\begin{aligned} 
&Z^3 \ell^{-2}t^3 &&\text{for\ \ } \ell\le Z^{-1/3},\\
&Z\ell^{-8}t^3    &&\text{for\ \ } \ell\ge Z^{-1/3}.
\end{aligned}\right.
\label{eqn-5.3}
\end{equation}

Consider the second term in (\ref{eqn-2.11}):
 \begin{align}
&(2\pi)^{-1} \int \sum_m Z_m \frac{(x-\x)\cdot (x-\y_m)}{|x-\x|\cdot |x-\y_m|^3}\rho_\Psi(x)\phi (|x-\x|)\,dx.
\notag\\
\intertext{We replace in the integral in the right-hand expression  $\rho_\Psi (x)$ by $\rho(x)$ and get}
&(2\pi)^{-1} \sum_m Z_m\int  \frac{(x-\x)\cdot (x-\y_m)}{|x-\x|\cdot |x-\y_m|^3}\rho (x)\phi (|x-\x|)\,dx
\label{eqn-5.4}\\
\shortintertext{with an error}
&(2\pi)^{-1} \sum_m Z_m \int  \frac{(x-\x)\cdot (x-\y_m)}{|x-\x|\cdot |x-\y_m|^3}(\rho_\Psi(x)- \rho (x))\phi (|x-\x|)\,dx
\label{eqn-5.5}
\end{align}
and one can see easily that (\ref{eqn-5.4}) does not exceed $CZ\zeta^3\ell^{-3}t^4$\,\footnote{Indeed, it suffices to take a half-sum of the integrand  in (\ref{eqn-5.4}) with its value at symmetric about $\x$ point, because both $|x-\y_m|$ and $\rho(x)$ satisfy $|\nabla f|\le Cf\ell^{-1}$.}.

To estimate (\ref{eqn-5.5}) we make a partition in $B(\x,t)$ with subelements supported in the layers 
$\{x\colon |x-\x|\asymp t'\}$ with $p< t'\le t$ and in $B(\x,p)$ with $\zeta^{-1}\le p \le t$. According to (\ref{eqn-3.12}) the contribution of each layer does not exceed 
$CZ \ell^{-2} \bigl(\zeta^2t'^2+ \zeta^{-2}Z^{5/3-\delta}\bigr)$ and summation over layers returns
its value as $t'=t$, with $\zeta^{-2}Z^{5/3-\delta}$ acquiring logarithmic factor with we compensate by decreasing $\delta$: 
\begin{equation}
CZ \ell^{-2} \bigl(\zeta^2t^2+ \zeta^{-2}Z^{5/3-\delta}\bigr).
\label{eqn-5.6}
\end{equation}

Meanwhile, contribution of the ball $B(\x,p)$  into (\ref{eqn-5.5}) does not exceed
$CZ \ell^{-2} \|\rho_\Psi-\rho\|_{\sL^1 (B(\x,p))}$ and to estimate it we use Theorem~\ref{elden2-thm-1.1} of \cite{ivrii:el-den-2} with $a=\ell$ and $\upmu=p^3\ell^{-3}$:
\begin{multline}
\|\rho_\Psi - \rho^\TF \|_{\sL^1(B(\x,s))}  \\[4pt]
\le C\left\{\begin{aligned}
&p^2   \ell^{-2/3}Z^{11/9-\delta/3} +p^3 Z\ell ^{-2}  
&&\text{as\ \ } p  \ge   Z^{-5/18-\delta}\ell^{5/6}, &&Z^{-1}\le \ell \le Z^{-1/3},\\
&p^2\ell^{-8/3}   Z^{5/9-\delta/3}    &&\text{as\ \ } p  \ge  Z^{5/9-\delta}\ell^{10/3}, &&Z^{-1/3}\le \ell \le Z^{-5/21}.
\end{aligned}\right.
\label{eqn-5.7}
\end{multline}

Therefore (\ref{eqn-2.11}) implies
\begin{multline}
\rho_\Psi (\x) \le 
C\Bigl(Z\zeta^3 \ell^{-3}t^4 + Z^{19/9}t^{-1} + Z\zeta^2\ell^{-2} t^2\Bigr)\\
+ C\Bigl(\zeta^3 +Z^{8/3-\delta}\zeta^{-2}\ell^{-2}+ Z\ell^{-2}\|\rho_\Psi-\rho\|_{\sL^1(B(\x,p))}\Bigr) \,,
\label{eqn-5.8}
 \end{multline}
where only first line depends on $t$. One can see easily that the third term in the first line does not exceed the sum of two first terms. Further,  the second term there is larger than $Z^{19/9}\ell^{-1}$ which is larger than $CZ^3$ as $\ell\le Z^{-8/9}$ and since we already have an estimate (\ref{eqn-2.10}), we should consider only $\ell\ge Z^{-8/9}$. Furthermore, $Z^{19/9}\ell^{-1}\ge \zeta^3$.

Finally, optimizing remaining two terms in the first line of (\ref{eqn-5.8}) by $t\colon \zeta^{-1}\le t\le \ell$, we get
\begin{multline}
\rho_\Psi (\x) \le 
C\Bigl( Z^{17/9} \zeta^{3/5} \ell^{-3/5}  + Z\zeta^{-1} \ell^{-3}   + Z^{19/9}\ell^{-1}\Bigr)\\
+ C\Bigl(Z^{8/3-\delta}\zeta^{-2}\ell^{-2}+ Z\ell^{-2}\|\rho_\Psi-\rho\|_{\sL^1(B(\x,p))}\Bigr) . 
\label{eqn-5.9}
 \end{multline}
Let us compare terms there.

\begin{enumerate}[label=(\roman*), wide, labelindent=0pt]
\item\label{pf-1.1-i}
Let $Z^{-8/9}\le \ell \le Z^{-1/3}$. Then one can see easily that the first line is defined by the third term $Z^{19/9}\ell^{-1}$ for $Z^{-8/9}\le \ell\le Z^{-7/9}$  and by the first term, which is $Z^{197/90}\ell^{-9/10}$, for $Z^{-7/9}\le \ell\le Z^{-1/3}$.

One can see easily that the first term in the second line of (\ref{eqn-5.9}) is smaller than $Z\ell^{-1}$. Using the first case in (\ref{eqn-5.7}) with $p=\max(\zeta^{-1}, \ell^{5/6}Z^{-5/18})$ and $\delta=0$, we see that  the second  term in the second line is smaller than the first line as well. Thus we arrive to Theorem~\ref{thm-1.1}, Statement~\ref{thm-1.1-i}.

\item\label{pf-1.1-ii}
Let $ \ell \ge Z^{-1/3}$. Then one can see easily that the first line of (\ref{eqn-5.9}) is defined by 
the first term, which is $Z^{17/9}\ell^{-9/5}$ for $Z^{-1/3}\le \ell\le Z^{-5/18}$ and by
$Z\ell^{-1}$ for $ \ell\ge Z^{-5/18}$.

Consider the second line and impose condition $\ell \le Z^{-2/9}$.  Then the first line dominates the first term here.  Using the second case in (\ref{eqn-5.7}) with $p=\max(\zeta^{-1}, \ell^{10/3}Z^{5/9})$ and $\delta=0$, we see that the first line dominates the last term in the second line as well.

Further, let $\ell \ge Z^{-2/9}$. Recall that the second line (except $C\zeta^3$) was a result of the estimate of the second term in (\ref{eqn-2.11}), which, however, could be estimated by 
\begin{equation}
CZ\ell^{-2}\int_{B(\x,\ell(\x))} \rho_\Psi (x)\,dx.
\label{eqn-5.10}
\end{equation}
It is well known that $\int \rho_\Psi \le CZ$ and therefore (\ref{eqn-5.10}) does not exceed $CZ^2\ell^{-2}$ which covers $\ell \ge Z^{-1/9}$.

Furthermore, in the remaining range $Z^{-2/9}\le \ell \le Z^{-1/9}$  we can use Proposition~\ref{prop-3.2}\ref{prop-3.2-ii}  to show, that the first term in the second line does not exceed $Z^{19/9}\ell^{-1}$ while the second term there is estimated again by 
the second case in (\ref{eqn-5.7}). Thus we arrive to Theorem~\ref{thm-1.1}, Statement~\ref{thm-1.1-ii}.

\item\label{pf-1.1-iii}
Finally, using Proposition~\ref{prop-3.2}\ref{prop-3.2-iii} we  prove Theorem~\ref{thm-1.1},  Statement~\ref{thm-1.1-iii}.
\end{enumerate}

\end{document}